\documentclass[a4paper]{llncs}

\usepackage[utf8]{inputenc}
\usepackage{amsmath}
\usepackage{amssymb}
\usepackage{bussproofs}
\usepackage{xfrac}
\usepackage{mathrsfs}
\usepackage{verbatim}
\usepackage{color}
\usepackage[dvipsnames]{xcolor}
\usepackage[colorlinks]{hyperref}
\usepackage{qtree}
\usepackage{textcomp}
\usepackage[linguistics]{forest}

\usepackage{xspace} 

\newcommand{\Formula}[1]{\ensuremath{\mathrm{formula}(#1)}}

\newcommand{\Literal}[3]{\ensuremath{\langle #1,#2,#3\rangle}}

\newcommand{\tuples}[1]{\ensuremath{\bar{#1}}}

\newcommand{\isdef}{\stackrel{\mbox{\tiny def}}{=}}   
\newcommand{\id}{\mathrm{\it id}}
\newcommand{\dom}{\mathrm{dom}}

\newcommand{\false}{\mathrm{\bf false}}
\newcommand{\true}{\mathrm{\bf true}}

\newcommand{\mliteral}{$M$-literal\xspace}  
\newcommand{\inliteral}{initial \mliteral}
\newcommand{\mclause}{$M$-clause\xspace}
\newcommand{\Mclauses}{$M$-Clauses}
\newcommand{\wfrm}{well-formed\xspace}   

\newcommand{\ordinary}[1]{{\tt #1}}

\newcommand{\lit}{l}                     
\newcommand{\litB}{m}
\newcommand{\litC}{k}
\newcommand{\cl}{C}                      
\newcommand{\clset}{\mathcal{C}}         
\newcommand{\oclset}{\ordinary{C}}       
\newcommand{\ocl}{\ordinary{\cl}}        
\newcommand{\form}{\ordinary{F}}         

\newcommand{\evar}{\alpha}               
\newcommand{\evarB}{\beta}

\newcommand{\tab}{{\mathcal T}}          
\newcommand{\otab}{\ordinary{T}}
\newcommand{\node}{\nu}                  
\newcommand{\nodeB}{\mu}
\newcommand{\leafA}{\lambda}             
\newcommand{\pathA}{P}                   
\newcommand{\abranch}{B}

\newcommand{\tuple}{\tuples{t}}
\newcommand{\tupleA}{\tuples{r}}

\newcommand{\pred}{L}                    
\newcommand{\predA}{P}                   
\newcommand{\predB}{Q}

\renewcommand{\S}{{\cal S}}              

\newcommand{\Vars}{{\cal V}}  

\newcommand{\uVars}{{\cal V}_{o}}  
\newcommand{\eVars}{{\cal V}_{a}}  
\newcommand{\UV}[1]{\ensuremath{\uVars (#1)}}
\newcommand{\EV}[1]{\ensuremath{\eVars(#1)}}

\newcommand{\VarsIn}[1]{\vars(#1)}

\newcommand{\idset}{\{ \id \}}

\newcommand{\keyprop}{($\star$)\xspace}
\newcommand{\subS}[1]{\eta_{#1}}

\newcommand{\conflict}{conflict\xspace}
\newcommand{\compact}{compact\xspace}
\newcommand{\mgu}{\text{mgu}}
\newcommand{\blocking}{blocking\xspace}
\newcommand{\domev}{\dom_{a}}
\newcommand{\esubst}{abstraction substitution\xspace}
\newcommand{\realizable}{realizable\xspace}

\newcommand{\bundledNF}{bundled normal form\xspace}
\newcommand{\BNF}{\ensuremath{\mathsf{BNF}}\xspace}

\newcommand{\vars}{{\cal V}} 
\newcommand{\evariable}{abstraction variable\xspace}

\newcommand\ident{\simeq}
\newcommand{\tableau}{$M$-tableau\xspace}

\title{A Tableaux Calculus for Reducing Proof Size}
\author{Michael Peter Lettmann\inst{1} \and Nicolas Peltier\inst{2}}
\institute{Technische Universit\"{a}t Wien, Institute of Information Systems, Vienna, Austria \email{michael.lettmann@tuwien.ac.at} \and Univ. Grenoble Alpes, CNRS, Grenoble INP, LIG, F-38000 Grenoble France \email{Nicolas.Peltier@univ-grenoble-alpes.fr}}

\begin{document}

\maketitle

\begin{abstract}
A tableau calculus is proposed, based on a compressed representation of clauses, where
literals sharing a similar shape may be merged. The 
inferences applied on these literals are fused when possible, which reduces the size of the proof.
It is shown that the obtained proof procedure is sound, refutationally complete and
allows to reduce the size of the tableau by an exponential factor. The approach is compatible with all usual refinements of tableaux. 
\end{abstract}

\section{Introduction}

Tableau methods (see for instance \cite{journals/jolli/Rijke01} or \cite{Haehnle:HandbookAR:tableaux:2001}) always played a crucial role in the development of new techniques for automated theorem proving. They are easy to comprehend and implement, well-adapted to interactive theorem proving, and, therefore, normally form the basis of the first proof procedure for any newly defined logic \cite{fit90}. Nonetheless, they cannot compete with resolution-based calculi both in terms of efficiency and deductive power (i.e.\ proof length, see for instance \cite{Eder92}).
This is partly due to the ability of resolution-based methods to generate lemmas and to simulate atomic cuts\footnote{We recall that the cut rule consists in expanding a
tableau by adding two branches with $\neg \phi$ and $\phi$ respectively, where $\phi$ is any formula (intuitively $\phi$ can be viewed as a lemma). A cut is atomic if $\phi$ is atomic.}  in a feasible way. There have been attempts to integrate some restricted forms of cut into tableau methods, improving both efficiency and proof size (see for instance \cite{letz94controlled,Haehnle:HandbookAR:tableaux:2001}). But, for more general forms of cuts, it is difficult to decide whether an application of the cut rule is useful or not, thus the rule is not really applicable during proof search. Instead, cuts may be introduced  after the proof is generated, to make it more compact by introducing lemmas and fusing recurring patterns \cite{DBLP:journals/tcs/HetzlLRW14,DBLP:journals/tcs/LeitschL18}.

In this paper, rather than trying to integrate cuts into the tableau calculus, we devise a new tableau procedure in which a proof compression, that is similar to the compressive power of a $\Pi_2$-cut, is achieved by employing a shared representation of literals. Formal definitions will be given later, but we now provide a simple example to illustrate our ideas. Consider the schema of clause sets: $\{ \bigvee_{i=1}^n p_0(a_i), \forall y.\neg p_n(y) \} \cup \{ \forall x.\neg p_{i-1}(x) \vee p_{i}(x) \mid i \in [1,n] \}$. A closed tableau can be constructed by adding $n$ copies of the clauses $\neg p_n(y^j)$ and $\neg p(x_{i-1}^j) \vee p(x_{i}^j)$ (for $i,j \in [1,n]$) and unifying all variables $x_i^j$ and $y^j$ with $a_j$. One gets a tableau of size $O(n^2)$. To make the proof more compact, we may merge the inferences applied for each $a_j$, since each of these constants are handled in the same way. This can be done by first applying the cut rule on the formula $\exists x.p(x)$. The branch corresponding the $\neg \exists x.p(x)$ can be closed by using the first clause. In the branch corresponding to $\exists x.p(x)$ a constant $c$ is generated by skolemization and the branch can be closed by unifying $x_i$ and $y$ with $c$. This yields a tableau of size $O(n)$. Since it is hard to guess in advance whether such an application of the cut rule will be useful or not, we investigate another solution allowing the same proof compression. We represent the disjunction $\bigvee_{i=1}^n p_0(a_i)$ by a single literal $p(\alpha)$, together with a set of substitutions $\{ [\alpha\backslash a_i] \mid i \in [1,n]\}$. Intuitively, this literal states that $p(\alpha)$ holds for some term $\alpha$, and the given set of substitutions specifies the possible values of $\alpha$. In the following, we call such variables $\alpha$ {\em {\evariable}s}. The clauses are kept as compact as possible by grouping all literals with the same heads and in some cases inferences may be performed uniformly regardless of the value of $\alpha$. In our example we get a tableau of size $O(n)$ by unifying $x_i^j$ and $y^j$ with $\alpha$, this tableau may be viewed as a compact representation of an ordinary tableau, obtained by making $n$ copies of the tree, with $\alpha = a_1,\dots,a_n$. If we find out that an inference is applicable only for some specific value(s) of $\alpha$ (e.g., if one wants to close a branch by unifying $p_0(\alpha)$ with a clause $\neg p_0(a_1)$), then one may ``separate'' the literal by isolating some substitution (or sets of substitutions) before proceeding with the inference.

In this paper, we formalize these ideas into a tableau calculus called {\em \tableau}. Basic inference rules are devised to construct {\tableau}x and a strategy is provided to apply these rules efficiently, keeping the tableau as compact as possible. We prove that the procedure is sound and refutationally complete and that it may reduce the size of the proofs by an exponential factor. Our approach may be combined with all the usual refinements of the tableau procedure.

\section{Notations}

We briefly review usual definitions (we refer to, e.g., \cite{DBLP:books/el/RobinsonV01} for details).
Terms, atoms and clauses are built as usual over a (finite) set of function symbols $\Sigma$ (including constants, i.e.\ nullary function symbols), an (infinite and countable) set of variables $\Vars$ and a (finite) set of predicate symbols $\Omega$. 
The set of variables occurring in an expression (term, atom or clause) $e$ is denoted by 
$\VarsIn{e}$. For readability, a term $f(t)$ is sometimes written $f t$. Ordinary (clausal) tableaux are trees labelled by literals and built by applying Expansion and Closure rules, the Expansion rule expands a leaf by $n$ children labelled by literals $l_1,\dots,l_n$, where a copy of $l_1 \vee \dots \vee l_n$ occurs in the clause set at hand, and the Closure rule closes a branch by unifying the atoms of two complementary literals. A substitution is a function (with finite domain) mapping variables to terms.  A substitution mapping $x_i$ to $t_i$ (for $i \in [1,n]$) is written $[ (x_1,\dots,x_n)\backslash (t_1,\dots,t_n)]$. The identity substitution (for $n=0$) is denoted by 
$\id$.  The image of an expression $e$ by a substitution $\sigma$ is defined  inductively as usual and written $e\sigma$.

\section{A Shared Representation of Literals}

We introduce the notion of an \mliteral, that is a compact representation of a disjunction of ordinary literals with the same shape. 
The interest of this representation is that it will allow us to perform similar inferences in parallel on all these literals. We assume that $\vars$ is partitioned into two (infinite) sets $\uVars$ and $\eVars$. The variables in $\uVars$ are ordinary variables. They may be either universally quantified variables in clauses, or rigid variables in tableaux. The variables in $\eVars$ are called {\em {\evariable}s}. These are not variables in the standard sense, but can been seen rather as placeholders for a term that may take different values in different literals or branches. These variables will permit to share inferences applied on different literals. The set of ordinary variables (resp.\ {\evariable}s) that occur in a term $Q$ is denoted by $\UV{Q}$ (resp.\ $\EV{Q}$). A {\em renaming} is  an injective substitution $\sigma$ such that 
$x \in \uVars \Rightarrow x\sigma \in \uVars$ and $\alpha \in \eVars \Rightarrow \alpha\sigma \in \eVars$.

\begin{definition}[Syntax of \Mclauses]
An \emph{\mliteral} is either $\true$ or a triple $\Literal{\pred}{\tuple}{\S}$, where:
\begin{itemize}
\item{$\pred$ is either a predicate symbol $\predA $ or the negation of a predicate symbol $\neg \predA $,}
\item{$\tuple$ is an $n$-tuple of terms, where $n$ is the arity of $\predA $,}
\item{and $\S$ is a set of substitutions $\sigma$ with the same domain $D \subseteq \EV{\tuple}$ (by convention $D$ is empty if $\S = \emptyset$) and such that  $\VarsIn{\tuple\sigma}\cap D=\emptyset$.}
\end{itemize}
An \emph{\mclause} is a set of {\mliteral}s, often written as a disjunction. 
\end{definition}

With a slight abuse of words, we will call the set $D$ in the above definition the {\em domain of $\S$} (denoted by $\dom(\S)$). 
The semantics of {\mclause}s is defined by associating each \mliteral with an ordinary clause (or $\true$):

\begin{definition}[Semantics of \Mclauses]
For every \mliteral $\lit$, we denote by $\Formula{\lit}$ the formula defined as follows (with the convention that empty disjunctions are equivalent to $\false$):
\[
\begin{array}{lll}
\Formula{\Literal{\pred}{\tuple}{\S}} & \isdef & \bigvee_{\theta\in \S} \pred(\tuple\theta) \\
\Formula{\true} & \isdef & \true \\
\end{array}
\]
For every \mclause $\cl$, we denote by $\Formula{\cl}$ the clause $\bigvee_{\lit \in \cl}\Formula{\lit}$.  
For every set of {\mclause}s $\clset$, we denote by $\Formula{\clset}$ the formula (in conjunctive normal form)
$\bigwedge_{\cl\in \clset}\Formula{\cl}$.

We write $E \ident E'$ iff $\Formula{E} = \Formula{E'}$ (up to the usual properties of $\vee$ and $\wedge$: associativity, commutativity and idempotence).
\end{definition}

\begin{example}
Let $\predA $ be a unary predicate, $\predB $ be a binary predicate, $c$ be a constant, $f,g$ be unary functions, $x$ be an ordinary variable, and $\alpha,\beta,\gamma$ be {\evariable}s. 
The triples $\lit_1 = \Literal{\predA }{\alpha}{\{ [\alpha\backslash f(c)]\}}$ and \[\lit_2= \Literal{\predB }{(\beta ,f(\gamma))}{\{ [(\beta ,\gamma )\backslash (f(c),c)],[(\beta ,\gamma )\backslash (c,f(c))]\}}\] are {\mliteral}s, and 
\[
\begin{array}{lll}
\Formula{\lit_1} & = & \predA (f(c)) \\
\Formula{\lit_2} & = & \predB (f(c),f(c))\lor \predB (c,f(f(c)))
\end{array}
\]
The common shape $\predB (\cdot,f(\cdot))$ is shared between the two literals  in the second clause.
\end{example}

\begin{remark}
Observe that if $\S = \emptyset$ then $\Formula{\Literal{\pred}{\tuple}{\S}} = \false$, i.e.\ $\Literal{\pred}{\tuple}{\S}$ denotes an empty clause.
Moreover, any ordinary literal  may be encoded as an \mliteral where the set of substitutions is a singleton, e.g., $\Literal{\predA }{(\evar,x)}{\{ [\evar\backslash a] \}} \simeq \predA (a,x)$. Also, an \mliteral $\Literal{\pred}{\tuple}{\{ \sigma \}}$ is always equivalent to $\Literal{\pred}{\tuple\sigma}{\{ \id \}}$. 
\end{remark}

The application of a substitution $\sigma$ to an \mliteral is defined as follows:
\[
\begin{array}{lll}
(\true)\sigma						&	\isdef	&	\true \\
\Literal{\pred}{\tuple}{\S}\sigma	&	\isdef	&	\Literal{\pred}{\tuple\sigma'}{\{ \theta\sigma \mid \theta \in \S \}} 
\end{array}
\]
where $\sigma'$ denotes the restriction of $\sigma$ to the variables not occurring in $\dom(\S)$.

\begin{example}
Let $\lit = \Literal{\predA}{(\alpha,x)}{ \{ [\alpha\backslash x], [\alpha\backslash y] \}}$ and $\sigma = [x\backslash a]$. Then: \[\lit\sigma = \Literal{\predA}{(\alpha,a)}{ \{ [\alpha\backslash a], [\alpha\backslash y] \}}\]

Let  $\lit' = \Literal{\predB}{(\alpha)}{ \{ [\alpha\backslash a], [\alpha\backslash y] \}}$ and $\theta = [\alpha\backslash a]$. Then $\lit\theta = \lit'$.
\end{example}

\begin{proposition}
Let $\lit = \Literal{\pred}{\tuple}{\S}$ be an \mliteral. If $\dom(\S) = \emptyset$, then one of the following conditions hold:
\begin{itemize}
\item{$\S = \emptyset$ and $\Formula{\lit} = \false$;}
\item{$\S = \{ \id \}$ and $\Formula{\lit} = \pred(\tuple)$.}
\end{itemize}
\end{proposition}
\begin{proof}
The identity is the only substitution with empty domain, whence the result.
\end{proof}

A given ordinary clause may be represented by many different {\mclause}s, for instance 
$\predA (a) \vee \predA (b)$ may be represented as 
$(\predA ,(a),\{ \id \}) \vee (\predA ,(b),\{ \id \})$ or
$(\predA ,(\evar),\{ [\evar\backslash a], [\evar\backslash b] \})$, or even $(\predA ,(\evar),\{ [\evar\backslash a], [\evar\backslash b] \}) \vee (\predB ,(\evarB),\emptyset)$. In practice it is preferable to start with a representation in which useless literals are deleted and in which the remaining literals are grouped when possible. This motivates the following:
\begin{definition}
An \mclause $\cl$ is in \emph{\bundledNF} (short: \BNF) if it satisfies the following conditions.
\begin{itemize}
\item{For every \mliteral $\Literal{\pred}{\tuple}{\S}\in \cl$, $\S \not = \emptyset$.}
\item{If $\true \in \cl$ then $\cl= \{ \true \}$.}
\item{For all distinct literals 
$\Literal{\pred_1}{\tuple_1}{\S_1},\Literal{\pred_2}{\tuple_2}{\S_2}\in C$, $\pred_1$ is distinct from $\pred_2$. 
}
\end{itemize}
A \mclause set $\mathcal{C}$ is in \BNF if all {\mclause}s of $\mathcal{C}$ are in \BNF.
\end{definition}

\begin{example}
\label{example:M-clauses in compressed normal form}
Let 
\[
\lit_1\isdef \Literal{\predA }{\alpha}{\{ [\alpha\backslash f(c)]\}},
\] 
\[
\lit_2\isdef \Literal{\predA }{\beta}{\{ [\beta\backslash f(c)]\}},
\]
\[
\lit_3\isdef\Literal{\predB }{(\beta ,f\gamma)}{\{ [(\beta ,\gamma )\backslash (f(c),c)],[(\beta ,\gamma )\backslash (c,f(c))]\}},
\]
and
\[
\lit_4\isdef \Literal{\predB }{(\beta ,\gamma)}{\{ [(\beta ,\gamma )\backslash (f(c),c)],[(\beta ,\gamma )\backslash (c,f(c))]\}}
\]
be {\mliteral}s. The \mclause $\{ \lit_3,\lit_4\}$ is not in \BNF while the {\mclause}s $\{ \lit_1,\lit_4\}$ and $\{\lit_2,\lit_4\}$ are in \BNF. 
\end{example}

\begin{definition}
An \mclause $\cl$ is {\em \wfrm} if for all distinct literals 
$\lit =\Literal{\pred_1}{\tuple_1}{\S_1}$ and $\litB =\Literal{\pred_2}{\tuple_2}{\S_2}$ in $C$, $\dom(\S_1) \cap \dom(\S_2) = \emptyset$.
\end{definition}

\begin{example}
Consider the two {\mclause}s $\cl_1:=\{ \lit_1,\lit_4\}$ and $\cl_2:=\{\lit_2,\lit_4\}$ of Example \ref{example:M-clauses in compressed normal form}. $\cl_1$ is \wfrm, $\cl_2$ is not \wfrm. By renaming, $\cl_2$ can be transformed into $\cl_1$.
\end{example}

It is clear that every \mclause can be transformed into an equivalent \wfrm \mclause by renaming. In the following, we shall implicitly assume that all the considered {\mclause}s are \wfrm.

\begin{lemma}
Let $\form$ be a formula in conjunctive normal form. Then there is an \mclause set $\clset$ in \BNF such that $\Formula{\clset}\ident \form$.
\end{lemma}

\begin{proof}
Let $\ocl$ be a clause of $\form$ and let $\pred_1,\ldots ,\pred_m$ the pairwise different predicate symbols or negated predicate symbols occurring in $\ocl$. For each symbol $\pred_i$, we collect all term tuples $\tuple_{1,i},\ldots ,\tuple_{k_i,i}$ such that $\pred_i(\tuple_{j,i}) \in \ocl$ (for $j = 1,\dots,k_i$). It is clear that all the tuples $\tuple_{j,i}$ for $j=1,\dots,k_i$ have the same length $n_i$, where $n_i$ is the arity of the predicate symbol of $\pred_i$. We define the \mliteral $\lit_i \isdef \Literal{\pred_i}{\tuples{\evar_i}}{\S_i}$ where $\tuples{\evar_i}$ is a tuple of $n_i$ fresh {\evariable}s and $\S_i$ is the set of substitutions $\{ [\tuples{\evar_i}\backslash\tuple_{1,i}],\ldots ,[\tuples{\evar_i}\backslash\tuple_{k_i,i}]\} $. Then we can define the \mclause $\cl\isdef \{ \lit_1,\ldots, \lit_m\}$. It is easy to check that $\Formula{\cl} = \ocl$ and that $\cl$ is in \BNF. By applying this method to every clause  of $\form$, we eventually get a set of {\mclause}s in \BNF $\clset$ such that $\Formula{\clset}\ident \form$.
\end{proof}

\begin{example}
Consider the clause $\{ \lit_3,\lit_4\}$ of Example \ref{example:M-clauses in compressed normal form}. It can be written in \BNF as 
$\Literal{\predB }{(\beta ,\gamma)}{\S}$, where $\S$ denotes the following set of substitutions:
\[
\{ [(\beta ,\gamma )\backslash (f(c),f(c))],[(\beta ,\gamma )\backslash (c,f(f(c)))], [(\beta ,\gamma )\backslash (f(c),c)],[(\beta ,\gamma )\backslash (c,f(c))]\}
\]
\end{example}

\section{A Tableaux Calculus for \Mclauses}

\label{sect:rules}

In this section, we devise a tableaux calculus for refuting sets of {\mclause}s. 
This calculus is defined by a set of inference rules, that, given an existing tableau $\tab$, allow one to:
\begin{enumerate}
\item{Expand a branch with new children, by introducing a new copy of an \mclause of the set at hand.}
\item{Instantiate some of the (rigid) variables occurring in the tableau.}
\item{Separate shared literals inside an \mclause, so that different inferences can be applied on each of the corresponding branches. The rule can be applied on nodes that are not leaves.}
\end{enumerate}
Steps $1$ and $2$ are standard, but Step $3$ is original.

\begin{definition}[Pre-Tableau]
\label{def:pre-tab}
A {\em pre-tableau} is  a tree $\tab$ where vertices are labelled by {\mliteral}s  or by $\false$. We call the direct successors of a node its \emph{children}. The {\em root} is the (unique) node that is not a child of any node in $\tab$ and a {\em leaf} is a node with no child. A \emph{path} $\pathA$ is a sequence of nodes $(\node_1,\ldots ,\node_n)$ such that $\node_{i+1}$ is a child of $\node_i$ for $i\in[1,n-1]$. Furthermore, we call $\node_1$ the \emph{initial node} of $\predA $ and $\node_n$ the \emph{last node} of $\pathA$. A \emph{branch} is a path such that the initial node is the root 
and the last node is a leaf. With a slight abuse of words we say that a branch contains an \mliteral $\lit$ if it contains
a node labelled by $\lit$.

The {\em descendants} of a node $\node$ are inductively defined as $\node$ and the descendants of the children of $\node$. The \emph{subtree of root $\node$ in $\tab$} is the subtree consisting of all the descendants of $\node$, as they appear in $\tab$.

If $\node$ is a non-leaf node with exactly $n > 0$ children $\node_1,\dots,\node_n$ labelled by {\mliteral}s $\lit_1,\dots,\lit_n$ respectively, then 
the {\em formula associated with $\node$} is defined as: $\bigvee_{i=1}^n \Formula{\lit_i}$.

We say that an \mliteral $\Literal{\pred}{\tuple}{\S}$  (resp.\ a node $\node$ labelled by $\Literal{\pred}{\tuple}{\S}$) \emph{introduces an \evariable $x$}  if $x \in \dom(\S)$.
\end{definition}

\begin{definition}
Let $\tab$ be a pre-tableau and $\sigma$ be a substitution.  Then $\tab\sigma$ denotes the result of applying $\sigma$ to all {\mliteral}s labelling the nodes of $\tab$. 
\end{definition}

\begin{definition}[Tableau for a Set of \Mclauses]
An {\em \tableau} $\tab$ for a set of {\mclause}s $\clset$ is a pre-tableau build inductively by applying the rules Expansion, Instantiation and Separation to an initial tableau containing only one node, labelled by $\true$ (also called the \emph{\inliteral}). 
\end{definition} 
In the following, the word ``tableau'' always refers to an {\tableau}, unless specified otherwise (we use the expression ``ordinary tableau" for standard ones).

The rules are defined as follows (in each case, $\tab$ denotes a previously constructed tableau for a set of {\mclause}s $\clset$.

\paragraph*{Expansion Rule.}

Let $\leafA$ be a leaf of $\tab$, and $\cl$ be an element of $\clset$ not containing  $\true$. Let $\cl'$ be a copy of $\cl$ where all variables that occur also in $\tab$ are renamed such that $\cl'$ share no variable\footnote{Note that both ordinary and {\evariable}s are renamed.} with $\tab$. The pre-tableau $\tab'$ constructed by adding a new child labelled by $\lit$ to $\leafA$ for each $\lit\in \cl'$ is a tableau for $\clset$.

\paragraph*{Instantiation Rule.}

Let $t$ be a term and $x\in\uVars$ such that for all nodes $\node, \node'$, if $\node$ is labelled by an \mliteral containing $x$ and $\node'$ introduces an \evariable $\evar\in \EV{t}$, then $\node$ is a proper descendant of $\node'$. Then $\tab [x\backslash t]$ is a tableau for $\clset$.

\begin{remark}
Observe that if $t$ contains no {\evariable}s then the condition always holds, since no node $\node'$ satisfying the above property exists. 
In practice, the Instantiation rule should of course not be applied with arbitrary variable and term. Unification will be used instead to find the most general instantiations closing a branch. A formal definition will be given later (see Definition \ref{def:compact}).
\end{remark}

\begin{example}
Let
\begin{center}
\begin{minipage}{.08\textwidth}
\[
\tab_1\isdef
\]
\end{minipage}
\begin{minipage}{.26\textwidth}
\begin{center}
\begin{small}
\begin{forest}
[{$\mathbf{true}$} 
 [{$\Literal{\neg \predA }{x}{\{\id\}}$} 
  [{$\Literal{\predA}{\alpha}{\{ [\alpha\backslash a], [\alpha \backslash b] \}}$} ] ] ]
\end{forest}
\end{small}
\end{center}
\end{minipage}
\hspace{.1\textwidth}
\begin{minipage}{.08\textwidth}
\[
\tab_2\isdef
\]
\end{minipage}
\begin{minipage}{.26\textwidth}
\begin{center}
\begin{small}
\begin{forest}
[{$\mathbf{true}$} 
 [{$\Literal{\predA}{\alpha}{\{ [\alpha\backslash a], [\alpha \backslash b] \}}$} 
  [{$\Literal{\neg \predA }{x}{\{\id\}}$} ] ] ]
\end{forest}
\end{small}
\end{center}
\end{minipage}
\end{center}
be two tableaux for some set of {\mclause}s $\clset$. The pre-tableau $\tab_1[x\backslash \alpha]$ is not a tableau, because $x$ is substituted by a term containing an \evariable $\alpha$, and $x$ occurs above the literal introducing $\alpha$. On the other hand, $\tab_2[x\backslash\alpha]$ is a tableau.
\end{example}

\paragraph*{Separation Rule.}

The rule is illustrated in Figure \ref{fig.initial tree Rule 3} towards Figure \ref{fig.S_2 empty tree Rule 3 .2}. Let $\node$ be a non-leaf node of $\tab$. Let $\nodeB$ be a child of $\node$, labelled by $\lit =\Literal{\pred}{\tuple}{\S}$. Let $\tupleA = \tuple\theta$ be an instance of $\tuple$, with $\dom(\theta) = \dom(\S)$  and $\EV{\tuple\theta} \cap \dom(\S) = \emptyset$. Let $\S_1$ be the set of substitutions $\sigma \in \S$ such that there exists a substitution $\sigma'$ with $\tuple\sigma = \tupleA\sigma'$ and every variable in $\dom(\sigma')$ is an \evariable not occurring in $\tab$, and let $\S_2 \isdef \S\setminus \S_1$.  Assume that $\S_1 \not = \emptyset$. We define the new literal $\lit'\isdef \Literal{\pred}{\tupleA}{\{ \sigma' \mid \sigma \in \S_1 \}}$. The Separation rule is defined as follows:
\begin{enumerate}
  \item We apply the substitution $\theta$ to $\tab$\footnote{Actually, due to the above conditions, the variables in $\dom(\theta)$ only occur in the subtree of root $\nodeB$, hence $\theta$ only affects this subtree.} .
  \item We replace the label $\lit$ of $\nodeB$ by $\lit'$.
  \item We add a new child to the node $\node$, labelled by a literal $\Literal{\pred}{\tuple}{\S_2}$.
\end{enumerate}
Observe that if $\S_2 = \emptyset$ then $\Formula{\Literal{\pred}{\tuple}{\S_2}} = \false$ hence the third step may be omitted, since the added branch is unsatisfiable anyway. The rule does not apply if $\S_1$ is empty.
\begin{figure}
\begin{minipage}[t]{.3\textwidth}
\centering
\caption{The initial tree in the Separation rule.}
\label{fig.initial tree Rule 3}
\begin{center}
\begin{small}
\begin{forest}
[{$\cdot$}
 [{$\cdot$}, edge=dotted
  [{$\litB$}
   [{$\lit$}
    [{$\mathcal{T}_3$} ] ] 
   [{$\mathcal{T}_2$} ] ] 
  [{$\mathcal{T}_1$} ] ] ]
\end{forest}
\end{small}
\end{center}
where $\litB$ is the label of node $\node$ and $\mathcal{T}_1,\mathcal{T}_2$, and $\mathcal{T}_3$ are possibly empty subtrees.
\end{minipage}
\hspace{.03\textwidth}
\begin{minipage}[t]{.3\textwidth}
\centering
\caption{The tree after an application of the Separation rule.}
\label{fig.S_2 empty tree Rule 3}
\begin{center}
\begin{small}
\begin{forest}
[{$\cdot$}
 [{$\cdot$}, edge=dotted
  [{$\litB\theta$}
   [{$\lit'$} 
    [{$\mathcal{T}_3\theta$} ] ]
   [{$\Literal{\pred}{\tuple}{\S_2}$} ]
   [{$\mathcal{T}_2\theta$} ] ] 
  [{$\mathcal{T}_1\theta$} ] ] ]
\end{forest}
\end{small}
\end{center}
where $\litB$ is the label of node $\node$ and $\mathcal{T}_1,\mathcal{T}_2$, and $\mathcal{T}_3$ are possibly empty subtrees.
\end{minipage}
\hspace{.03\textwidth}
\begin{minipage}[t]{.3\textwidth}
\centering
\caption{The tree without redundant substitutions after an application of the Separation rule.}
\label{fig.S_2 empty tree Rule 3 .2}
\begin{center}
\begin{small}
\begin{forest}
[{$\cdot$}
 [{$\cdot$}, edge=dotted
  [{$\litB$}
   [{$\lit'$} 
    [{$\mathcal{T}_3\theta$} ] ]
   [{$\Literal{\pred}{\tuple}{\S_2}$} ]
   [{$\mathcal{T}_2$} ] ] 
  [{$\mathcal{T}_1$} ] ] ]
\end{forest}
\end{small}
\end{center}
where $\litB$ is the label of node $\node$ and $\mathcal{T}_1,\mathcal{T}_2$, and $\mathcal{T}_3$ are possibly empty subtrees.
\end{minipage}
\end{figure}

\begin{example}
\label{exa.rules tableau}
Let 
\begin{align*}
\lit_1:= &\; \Literal{\predA }{\alpha}{\{ [\alpha\backslash fc]\}}, \\ 
\lit_2:= &\; \Literal{\neg \predA }{\alpha'}{\{ [\alpha'\backslash fc]\}}, \\ 
\lit_3:= &\; \Literal{\neg \predB }{(\beta' ,\gamma')}{\{ [(\beta' ,\gamma' )\backslash (fx,y)]\}},\text{ and} \\
\lit_4:= &\; \Literal{\predB }{(\beta ,\gamma)}{\{ [(\beta ,\gamma )\backslash (fc,c)],[(\beta ,\gamma )\backslash (z,fc)],[(\beta ,\gamma )\backslash (fz,fc)]\}}
\end{align*}
be {\mliteral}s and $\mathcal{C}=\{\{ \lit_1,\lit_4\} ,\{ \lit_2\} ,\{\lit_3\}\}$ be an \mclause set in \BNF. Applying three times the Expansion rule, we can derive the tableau
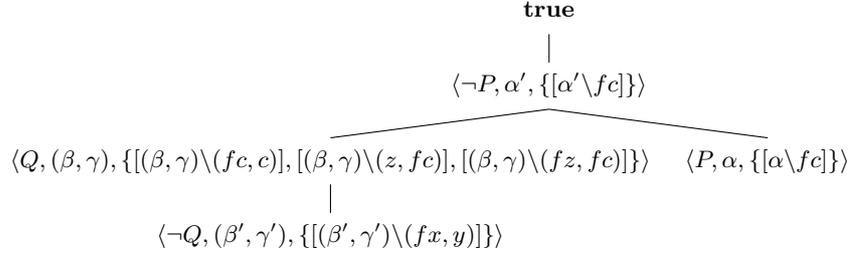
\begin{center}
\begin{small}
\begin{forest}
[{$\mathbf{true}$} 
 [{$\Literal{\neg \predA }{\alpha'}{\{ [\alpha'\backslash fc]\}}$} 
  [{$\Literal{\predB }{(\beta ,\gamma)}{\{ [(\beta ,\gamma )\backslash (fc,c)],[(\beta ,\gamma )\backslash (z,fc)],[(\beta ,\gamma )\backslash (fz,fc)]\}}$} 
   [{$\Literal{\neg \predB }{(\beta' ,\gamma')}{\{ [(\beta' ,\gamma' )\backslash (fx,y)]\}}$} ] ] 
  [{$\Literal{\predA }{\alpha}{\{ [\alpha\backslash fc]\}}$} ] ] ]
\end{forest}
\end{small}
\end{center}
Now, we apply two times the Separation rule. First we choose $\litB$ of Figure \ref{fig.initial tree Rule 3} to be $\mathbf{true}$ and $\lit$ to be $\Literal{\neg \predA }{\alpha'}{\{ [\alpha'\backslash fc]\}}$, where the substitution $\theta$ is $[\alpha'\backslash fc]$ (hence we get $\S_1 = \{ \id \}$). Afterwards, we choose analogously $\Literal{\neg \predA }{fc}{\{ \id \}}$ (which is the result of the first application) and $\Literal{\predA }{\alpha}{\{ [\alpha\backslash fc]\}}$, with the substitution $[\alpha\backslash fc]$. Both times, the tuple $\tupleA$ of the Separation rule is $fc$ (with $\S_2 = \emptyset$ in both cases). This leads to the tableau:
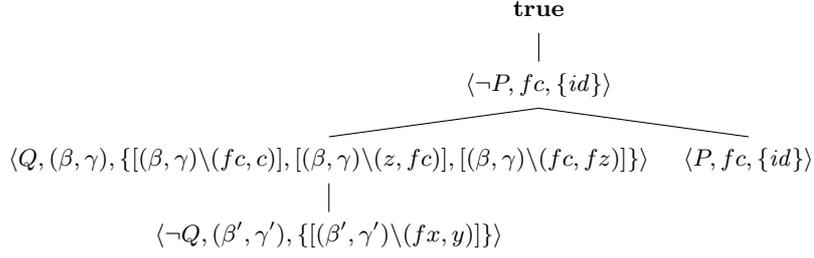
\begin{center}
\begin{small}
\begin{forest}
[{$\mathbf{true}$} 
 [{$\Literal{\neg \predA }{fc}{\{ \id \}}$} 
  [{$\Literal{\predB }{(\beta ,\gamma)}{\{ [(\beta ,\gamma )\backslash (fc,c)],[(\beta ,\gamma )\backslash (z,fc)],[(\beta ,\gamma )\backslash (fc,fz)]\}}$} 
   [{$\Literal{\neg \predB }{(\beta' ,\gamma')}{\{ [(\beta' ,\gamma' )\backslash (fx,y)]\}}$} ] ] 
  [{$\Literal{\predA }{fc}{\{ \id \}}$} ] ] ]
\end{forest}
\end{small}
\end{center}
Afterwards, we again apply the Separation rule, to modify the node labelled with $\lit_4$ 
where $\tupleA =(f\delta ,\gamma^*)$ and $\S_2=\{ [(\beta ,\gamma )\backslash (z,fc)]\}$. We abbreviate $\Literal{\predA }{fc}{\{ \id \}}$ by $\litB$.
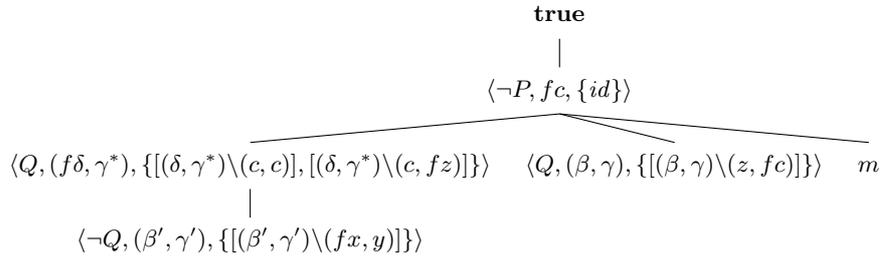
\begin{center}
\begin{small}
\begin{forest}
[{$\mathbf{true}$} 
 [{$\Literal{\neg \predA }{fc}{\{ \id \}}$} 
  [{$\Literal{\predB }{(f\delta ,\gamma^*)}{\{ [(\delta ,\gamma^* )\backslash (c,c)],[(\delta ,\gamma^* )\backslash (c,fz)]\}}$} 
   [{$\Literal{\neg \predB }{(\beta' ,\gamma')}{\{ [(\beta' ,\gamma' )\backslash (fx,y)]\}}$} ] ]
  [{$\Literal{\predB }{(\beta ,\gamma)}{\{ [(\beta ,\gamma )\backslash (z,fc)]\}}$} ]
  [{$\litB$} ] ] ]
\end{forest}
\end{small}
\end{center}
After some further applications of the Separation and Expansion rules, we are able to construct the following tableau by applying the Instantiation rule with the substitutions $[z\backslash fx'],[y\backslash fc],[x\backslash\delta],[y\backslash\gamma^*]$.
\begin{center}
\begin{small}
\begin{forest}
[{$\mathbf{true}$} 
 [{$\Literal{\neg \predA }{fc}{\{ \id \}}$} 
  [{$\Literal{\predB }{(f\delta ,\gamma^*)}{\{ [(\delta ,\gamma^* )\backslash (c,c)],[(\delta ,\gamma^* )\backslash (c,ffx')]\}}$} 
   [{$\Literal{\neg \predB }{(f\delta ,\gamma^*)}{\{ \id \}}$} ] ]
  [{$\Literal{\predB }{(fx',fc)}{\{ \id \}}$} 
   [{$\Literal{\neg \predB }{(fx',fc)}{\{ \id \}}$} ] ]
  [{$\Literal{\predA }{fc}{\{ \id \}}$} ] ] ]
\end{forest}
\end{small}
\end{center}
\end{example}

\section{Soundness}

\begin{definition}
Let $\tab$ be a pre-tableau or a tableau. A branch $\abranch$ of $\tab$ is \emph{closed} if it contains $\false$ or two nodes labelled by literals $\Literal{\pred_1}{\tuple_1}{\S_1},\Literal{\pred_2}{\tuple_2}{\S_2}$ such that $\tuple_1=\tuple_2$, $\pred_1=\predA $, $\pred_2=\neg \predA $ for some predicate symbol $\predA $. 
The (pre)-tableau $\tab$ is {\em closed} iff all branches of $\tab$ are closed.
\end{definition}
\begin{example}
The final tableau of Example \ref{exa.rules tableau} contains three branches, i.e.\ 
\begin{align*}
& \{\Literal{\neg \predA }{fc}{\emptyset},\Literal{\predB }{(f\delta ,\gamma^*)}{\{ [(\delta ,\gamma^* )\backslash (c,c)],[(\beta ,\gamma )\backslash (c,ffx')]\}},\Literal{\neg \predB }{(f\delta ,\gamma^*)}{\emptyset}\}, \\
& \{\Literal{\neg \predA }{fc}{\emptyset},\Literal{\predB }{(fx',fc)}{\emptyset},\Literal{\neg \predB }{(fx',fc)}{\emptyset}\}\text{, and} \\
& \{\Literal{\neg \predA }{fc}{\emptyset},\Literal{\predA }{fc}{\emptyset}\} .
\end{align*}
All of them are closed and so the tableau is closed. Observe that the inferences closing the branches corresponding to the literals $\predB(f(c),c)$ and $\predB(f(z),f(c))$ in $(\lit_1,\lit_4)$ are shared in the constructed tableau (both branches are closed by introduced suitable instances of $\lit_3$), whereas the literal $\predB(z,f(c))$ is handled separately (by instantiating $z$ by $f(x')$ and using yet another instance of $\lit_3$).
\end{example}

\begin{proposition}
\label{prop:disjointdom}
Let $\tab$ be a tableau. If $\node_1$ and $\node_2$ are distinct nodes
in $\tab$, labelled by the {\mliteral}s $\Literal{\pred_1}{\tuple_1}{\S_1}$ and $\Literal{\pred_2}{\tuple_2}{\S_2}$ respectively, then 
$\S_1$ and $\S_2$ have disjoint domains.
\end{proposition}
\begin{proof}
This is immediate since the {\mclause}s are  \wfrm and since variables are renamed before all applications of the Expansion rule, so that the considered \mclause share no variable with the tableau. Afterwards, none of the construction rules can affect the domain of the substitutions. For the Separation rule, the condition on $\dom(\sigma')$ in the definition of the rule guarantees that the new \mliteral $\lit'$ introduces variables that are distinct from those in $\tab$. The detailed proof is by an easy induction on tableaux.
\end{proof}

\begin{proposition}
\label{prop:cl_preserved}
Let $\tab$ be a tableau for a set of {\mclause}s $\clset$. For every non-leaf node $\node$ in $\tab$, the formula associated with $\node$ (as defined in Definition \ref{def:pre-tab}) is an instance of a formula $\Formula{\cl}$, where $\cl$ is a renaming of an \mclause in $\clset$. 
\end{proposition}
\begin{proof}
It suffices to show that all the construction rules preserve the desired property.
\begin{itemize}
 \item{{\bf Expansion.} The property immediately holds for the nodes on which the rule is applied, by definition of the rule. 
 The other nodes are not affected.}
 \item{{\bf Instantiation.} By definition, the formula associated with a node $\node$ in the final tableau is an instance of the formula associated with $\node$ in the initial one. Thus the property holds.}
 \item{{\bf Separation.} The nodes occurring outside of the subtree of root $\nodeB$ are not affected. By definition, the formula associated with the descendants of $\nodeB$ in the new tableau are instances of formulas associated with nodes of the initial tableau. Thus it only remains to consider the node $\node$. The formula associated with $\node$ in the final tableau is obtained from that of the initial one by removing the formula corresponding to an \mliteral $\lit =\Literal{\pred}{\tuple}{\S}$ and replacing it by $\Formula{\lit'} \vee \Formula{\Literal{\pred}{\tuple}{\S_2}}$. Since $\lit' = \Literal{\pred}{\tupleA}{\{ \sigma' \mid \sigma \in \S_1 \}}$, we have $\Formula{\lit'} =  \bigvee_{\sigma \in \S_1} \pred(\tupleA\sigma') = \bigvee_{\sigma \in \S_1} \pred(\tuple\sigma)$ (since  $\tuple\sigma = \tupleA\sigma'$ by definition of the Separation rule). Thus $\Formula{\lit'} \vee \Literal{\pred}{\tuple}{\S_2} = \Literal{\pred}{\tuple}{\S}$ and the proof is completed.}
\end{itemize}
\end{proof}

\begin{proposition}
\label{prop:vardependency}
Let $\tab$ be a tableau for $\clset$.
Let $\evar \in \eVars$ be a variable introduced in a node $\node$ and assume that $\evar$ occurs in an \mliteral labelling a node $\nodeB$.
Then $\nodeB$ is a descendant of $\node$.
\end{proposition}
\begin{proof}
By an easy induction on tableau. The property is preserved by any application of the 
Expansion rule because the variables in the considered \mclause are renamed by fresh variables, it is preserved by the Instantiation rule due to the conditions associated with the rule,  and the Separation rule only instantiates {\evariable}s.
\end{proof}

\begin{lemma}
\label{lem.lemma Soundness}
Let $\tab$ be a closed tableau for $\clset$ and let $\tab'$ be the tableau after applying once the Separation rule to a node $\node$ of $\tab$ with a child $\nodeB$ labelled with $\lit$. Then there is at most one branch $\abranch$ in $\tab'$ that is not closed. Moreover, this branch necessarily contains the  node  labelled by $\Literal{L}{\tuple}{\S_2}$ (see Figure \ref{fig.S_2 empty tree Rule 3 .2}). In particular, if $\S_2$ is empty then $\Literal{L}{\tuple}{\S_2}$  is $\false$ and $\tab'$ is closed.
\end{lemma}

\begin{proof}
Let $\abranch'$ be a branch in $\tab'$ and let $\theta$ be the substitution that we apply to $\tab$ in the Separation rule (see Figure \ref{fig.S_2 empty tree Rule 3 .2}). Assume that $\abranch'$ does not contain $\Literal{L}{\tuple}{\S_2}$. 
Then there is a branch $\abranch$ in $\tab$ where for every literal $\lit_i \in \abranch$ (with $i \in [1,n]$), either  $\lit_i = \lit$ and $\lit'\in \abranch'$, where $\lit'$ replaced $\lit$ during the application of the Separation rule, or $\lit_i\theta \in \abranch'$. Since $\tab$ is closed, $\abranch$ is closed, i.e.\ there are two {\mliteral}s $\litB,\litC\in\abranch$ such that $\litB =\Literal{\predA }{\tuple}{\S}$ and $\litC =\Literal{\neg \predA }{\tuple}{\S'}$. If $\litB\theta, \litC\theta \in \abranch'$, then (by Proposition \ref{prop:disjointdom}) $\theta \cap \dom(\S) = \theta \cap \dom(\S') = \emptyset$, thus $\litB\theta = \Literal{\predA}{\tuple\theta}{\{ \sigma\theta \mid \sigma \in \S \}}$ and $\litC\theta = \Literal{\neg \predA}{\tuple\theta}{\{ \sigma\theta \mid \sigma \in \S' \}}$. Consequently $\abranch'$ is closed. Otherwise, one of the literals $\litB$ or $\litC$ is $\lit$, say $\litB = \lit$, and $\litC\theta \in \abranch'$. Then, by definition of the Separation rule, we have $\lit' = \Literal{\predA}{\tupleA}{\{ \sigma' \mid \sigma \in \S_1\}} = \Literal{\predA}{\tuple\theta}{\{ \sigma' \mid \sigma \in \S_1\}}$.  By Proposition \ref{prop:disjointdom}, $\theta \cap \dom(\S') = \emptyset$, thus $\litC\theta = \Literal{\neg \predA}{\tuple\theta}{\{ \sigma\theta \mid \sigma \in \S' \}}$, hence $\abranch'$ is closed. 
\end{proof}

\begin{theorem}[Soundness]
\label{theo:sound}
If a set of {\mclause}s $\clset$ admits a closed tableau $\tab$ then $\clset$ is unsatisfiable.
\end{theorem}

\begin{proof}

We prove soundness by transforming $\tab$ into a closed tableau that contains only {\mliteral}s where the substitution set is $\{\id\}$. Due to Proposition \ref{prop:cl_preserved}, the resulting tableau then corresponds to an ordinary tableau, i.e.\ a tableau constructed only by the Expansion rule and the Instantiation rule. The {\mliteral}s could be replaced by usual literals. The soundness of ordinary tableau then implies the statement.

We transform the tableau $\tab$ by an iterative procedure. We always take the topmost node $\node$ labelled with an \mliteral $\lit =\Literal{\pred}{\tuple}{\S}$ where $\S\neq\{\id\}$ (if there is more than one topmost literal we can arbitrarily choose one). Then we consider a substitution $\theta\in \S$ and we apply the Separation rule with the tuple $\tuple\theta$. We have $\tuple\theta=\tuple\theta\id$ and if $\sigma\in \S \setminus \{ \theta \}$, then $\tuple\sigma\not = \tuple\theta$, hence 
there is no substitution $\sigma'$ with $\tuple\sigma\sigma' = \tuple\theta$, such that $\dom(\sigma')$ only contains  fresh variables. Consequently, the rule splits $\S$ into  a singleton $\S_1 = \{\theta\}$ and $\S_2=\S_1\setminus \{\theta\}$. The literal $\lit$ gets replaced by $\lit' =\Literal{\pred}{\tuple\theta}{\{\id\}}$ and we add the node $\nodeB$ labelled with $\Literal{\pred}{\tuple}{\S_2}$ (note that $\node$ cannot be the root of the tableau, because the root is always labelled by $\true$). By Lemma \ref{lem.lemma Soundness}, there is at most one non-closed branch, i.e.\ the branch ending with the node $\nodeB$ is the only open branch. We consider a copy $\tab_{\node}$ of the subtree of root $\node$ in $\tab$, renaming all variables introduced in $\tab_{\node}$ by fresh variables. We replace the root node of $\tab_{\node}$ by $\Literal{\pred}{\tuple}{\S_2}$ and replace the subtree of root $\nodeB$ in the tableau by $\tab_{\node}$. It is easy to check that the obtained tableau is a closed tableau for $\clset$. Furthermore, the length of the branches does not increase, the number of non-empty substitutions occurring in the {\mliteral}s does not increase, and it decreases strictly in $\lit'$ and $\Literal{\pred}{\tuple}{\S_2}$.
This implies that the multiset of multisets $\{ |\{ \sigma \in \S_1' \mid \sigma \not = \id \}|,\dots,|\{ \sigma\in \S_n' \mid \sigma\not=\id \}| \}$  of natural numbers, where $\{ \Literal{\pred_i}{\tuple_i}{\S_i'} \mid i \in [1,n] \}$ is a branch in $\tab$ is strictly decreasing according to the multiset extension of the usual ordering. Since this ordering is well-founded, the process eventually terminates, and after a finite number of applications of this procedure we get a tableau only containing nodes labelled with {\mliteral}s whose substitution set is equal to $\{\id\}$. 
This finishes the proof.
\end{proof}

\begin{remark}
\label{rem:bound}
As a by-product of the proof, we get that the size of the minimal ordinary tableau for a clause set $\Formula{\clset}$
is bounded exponentially by the size of any closed tableau for $\clset$. Indeed, we constructed an ordinary tableau from an \tableau $\tab$ in which every branch $(\lit_1,\dots,\lit_n)$ in $\tab$ is replaced by 
(at most) $k^n$ branches, where $k$ is the maximal number of substitutions in $\lit_i$.
In Section \ref{sect:exp}, we shall prove that this bound is precise, i.e.\ that our tableau calculus allows exponential reduction of proof size w.r.t.\ ordinary (cut-free) tableaux.
\end{remark}

\section{Completeness}

Proving completeness of \tableau is actually a trivial task, since one could always apply the Separation rule in a systematic way on all {\mliteral}s to transform them into ordinary literals (as it is done in the proof of Theorem \ref{theo:sound}), and then get the desired result by completeness of ordinary tableau. However, this strategy would not be of practical use. Instead, we shall devise a strategy that keeps the {\tableau} as compact as possible and at the same time allows one to ``simulate'' any application of the ordinary expansion rules. In this strategy, the Separation rule is applied on demand, i.e.\ only when it is necessary to close a branch. No hypothesis is assumed on the application of the ordinary expansion rules, therefore the proposed strategy is ``orthogonal'' to the usual refinements of ordinary tableaux, for instance connection tableaux\footnote{Connection tableaux can be seen as ordinary tableaux in which any application of the Expansion rule must be followed by the closure of a branch, using one of the newly added literals and the previous literal in the branch.} \cite{778537} or hyper-tableaux\footnote{Hyper-tableaux may be viewed in our framework as ordinary tableaux in which the Expansion rule must be followed by the closure of all the newly added branches containing negative literals.} \cite{Baumgartner:HyperNextGeneration:Tableaux:98}. Thus our approach can be combined with any refutationally complete tableau procedure \cite{Haehnle:HandbookAR:tableaux:2001}.

The main idea denoted by \emph{simulate a strategy} is to do the same steps as in ordinary tableau, while keeping {\mclause}s as compressed as possible. If ordinary tableau expands the tableau by a clause, we expand the tableau with the corresponding \mclause, and if a branch is closed in the ordinary tableau, then the corresponding branch is closed in the \tableau. This last step is not trivial:
Given two ordinary literals $\predA(\tuple)$ and $\neg \predA(\tupleA)$ the ordinary tableau might compute the most general unifier (mgu) of $\tuple$ and $\tupleA$. But in the presented formalism, the two literals might not appear as such, i.e.\ there are no literals $\litB=\Literal{\predA}{\tuple}{\{\id\}}$ and $\litC=\Literal{\neg\predA}{\tupleA}{\{\id\}}$ . In general, there are only {\mliteral}s $\litB' =\Literal{\predA}{\tuple'}{\S_1}$ and $\litC' =\Literal{\neg\predA}{\tupleA'}{\S_2}$ such that $\tuple'\theta = \tuple$ and $\tupleA'\vartheta = \tupleA$, where $\theta$ and $\vartheta$ denote the compositions of the substitutions occurring in the {\mliteral}s in the considered  branch. Note also that, although  $\tuple'$ and $\tupleA'$ are unifiable, the Instantiation rule cannot always be applied to unify them and close the branch. Indeed, the domain of the mgu may contain {\evariable}s, whereas the Instantiation rule only handles universal variables. For showing completeness, it would suffice to apply the Separation rule on each ancestor of $\litC'$ and $\litB'$ involved in the definition of $\theta$ or $\vartheta$, to create a branch where the literals $\litB$ and $\litC$ appear explicitly. Thereby, we would loose a lot of the formalisms benefit. Instead, we shall introduce a strategy that uses the Separation rule only if this is necessary for making the unification of $\tuple'$ and $\tupleA'$ feasible (by mean of the Instantiation rule). Such applications of the Separation rule may be seen as preliminary steps for the Instantiation rule. This follows the maxim to stay as general as possible because a more general proof might be more compact.

In the formalisation of the Instantiation rule we ensured soundness by allowing {\evariable}s only to occur in descendants of the literal that introduced the variable. This has a drawback to our strategy: The unification process that we try to simulate can ask for an application of the Instantiation rule which would cause a violation of this condition for {\evariable}s if we follow the procedure in the former paragraph. We thus have to add further applications of the Separation rule to
ensure that this condition is fulfilled.

\begin{definition}
\label{def:conflict}
Let $\abranch = (\node_0,\node_1,\dots,\node_n)$ be a path in a tableau $\tab$ where $\node_0$ is the initial node of $\tab$ and each node $\node_i$ (with $i > 0)$ is labelled by $\Literal{\pred_i}{\tuple_i}{\S_i}$.

An {\em \esubst} for $\abranch$ is a substitution $\eta_n\dots\eta_1$ with $\eta_i \in \S_i$, for $i=1,\dots,n$.

A {\em \conflict} in a branch $\abranch$ is a pair $(\tuple_i,\tuple_j)$ with $i,j \in [1,n]$, $\pred_i$ and $\pred_j$ are dual and $\tuple_i$ and $\tuple_j$ are unifiable. A \conflict is {\em $\eta$-\realizable} if $\eta$ is an  \esubst  for $\abranch$ such that $\tuple_i\eta$ and $\tuple_j\eta$ are unifiable. 
\end{definition}

In practice, we do not have to check that a \conflict is \realizable (this would be costly since we have to consider exponentially many substitutions).

If $(\tuple,\tuple')$ is a \conflict then $\tuple$ and $\tuple'$ are necessarily unifiable, with some mgu $\theta$. As mentioned before, this does not mean that a branch with \conflict can be closed. Moreover, according to the restriction on the Instantiation rule, a variable $x$ cannot be instantiated by a term containing an \evariable $\evar$, if $x$ occurs in some ancestor of the literal introducing $\evar$ in the tableau. This motivates the following:

\begin{definition}
A variable $\evar\in\eVars$ is {\em \blocking} for a \conflict $(\tuple,\tuple')$, where $\theta=\mgu(\tuple,\tuple')$ if 
 $\evar\in \dom(\theta)$ or $\evar$ occurs in a term $x\theta$, where $x \in \uVars$ and $x$ occurs in a literal labelling an ancestor of the node introducing $\evar$.
\end{definition}

Finally, we introduce a specific application of the Separation rule 
which allows one either to ``isolate'' some literals in order to ensure that they have a specific ``shape'' (as specified by a substitution), or to eliminate {\evariable}s completely if needed.

\begin{definition}
\label{def:compact}
If $\sigma$ is a substitution, we denote by $\domev(\sigma)$ the set of variables $\evar\in \eVars$ such that
$\evar\sigma\not \in \eVars$.

A tableau is {\em \compact} if it is constructed by a sequence of applications of the tableau rule in which:

\begin{itemize}
\item{The Instantiation rule is applied only if the tableau contains a branch with a \conflict $(\tuple,\tuple')$, with no \blocking variable. Each variable $x\in \dom(\sigma)$ is replaced by a term $t\sigma$, where $\sigma=\mgu(\tuple,\tuple')$ (since there is no \blocking variable it is easy to check that the conditions on the Instantiation rule are satisfied). Afterwards, it is clear that the branch is closed.}

\item{The Separation rule is  applied on a node labelled by $\Literal{\pred}{\tuple}{\S}$, using a substitution $\theta$  only if there exists a \conflict $(\tuple,\tuple')$ with $\sigma = \mgu(\tuple,\tuple')$ such that one of the following conditions holds:
\begin{enumerate}
\item{$\dom(\S)$ contains a \blocking variable in $\domev(\sigma)$ and $\theta$ is defined as follows: $\dom(\theta) = \dom(\S) \cap \domev(\sigma)$, $x\theta=x\sigma$ if $x\sigma$ is a variable, otherwise $x\theta$ is obtained from $x\sigma$ by replacing all variables by pairwise distinct fresh {\evariable}s. \label{sep:inst}}
\item{Or $\dom(\S)$ contains a \blocking variable not occurring in $\domev(\sigma)$, and $\theta \in \S$. \label{sep:elim}}
\end{enumerate}}
\end{itemize}
\end{definition}

The applications of the Separation rule in Definition \ref{def:compact}
are targeted at making the closure of the branch possible by getting rid of \blocking variables, while keeping the tableau as compact as possible (thus useless separations are avoided).

\begin{example}
For instance, assume that we want to close a branch containing two literals 
$\lit = \Literal{\pred}{(\evar)}{\{[\evar\backslash f(a)],[\evar\backslash f(b)],[\evar\backslash a]\}}$
and
$\lit' = \Literal{\neg \pred}{f(x)}{\{ \id \}}$. To this aim, we need to ensure that $\evar$ is unifiable with $f(x)$. This is done by
applying the separation rule (Case \ref{sep:inst}  of Definition \ref{def:compact}) with the substitution $[\evar\backslash f(\evarB)]$, so that $\evar$ has the desired shape. This yields: $\Literal{\pred}{(f(\evarB))}{\{[\evarB\backslash a],[\evarB\backslash b]\}}$
Afterwards, if $x$ does not occur before $\lit$ in the branch then the branch is closed by unifying $x$ with $\evarB$.
If $x$ occurs before $\lit$, then this is not feasible since this would contradict the condition on the Instantiation rule, and we have to apply the Separation rule  again  (Case \ref{sep:elim}) to eliminate $\evarB$, yielding (for instance) $\Literal{\pred}{(f(a))}{\{ \id \}}$. The direct application of the Separation rule with, e.g., $\theta = [\evar\backslash f(a)]$ is forbidden in the strategy.
\end{example}

\begin{theorem}[Completeness]
Let $\oclset$ be a clause set and $\clset$ be an \mclause set with $\Formula{\clset}\ident\oclset$. If $\oclset$ is unsatisfiable, then there is a closed \compact tableau for $\clset$.
\end{theorem}
\begin{proof}
For an \mclause $\cl_i\in\clset$, we denote by $\ordinary{\cl}_i$ the ordinary clause in $\oclset$ with $\Formula{\ordinary{\cl}_i}\simeq\Formula{\cl_i}$. W.l.o.g.\ we can assume that clauses do not appear twice, neither in $\clset$ nor in $\oclset$. Now, we can perform a proof search based on ordinary tableau starting with $\oclset$ (using any complete strategy) yielding a closed tableau, built by applying the usual Expansion and Closure rules. In the following,  $\otab$ will denote an already constructed ordinary tableau for $\oclset$ and $\tab$ will represent the corresponding \tableau for $\clset$. More precisely, the tableau $\tab$ is constructed in such a way that there exists an injective mapping $h$ from the nodes in $\tab$ to those in $\otab$, a function $\node \mapsto \eta_\node$ mapping each node $\node$ in $\tab$ labelled by some literal $\Literal{\pred}{\tuple}{\S}$ to an \esubst $\eta$  and a substitution $\vartheta$ (with $\dom(\vartheta) \subseteq \uVars$)  such that the following property holds (denoted by \keyprop):
\begin{enumerate}
\item{The root of $\tab$ is mapped to the root of $\otab$.}
\item{If $\node$ is a child of $\node'$ then $h(\node)$ is a child of $h(\node')$.}
\item{For any path $(\node_0,\dots,\node_n)$, if $\node_n$ is labelled by an \mliteral $\Literal{\pred}{\tuple}{\S}$, then $h(\node_n)$ is labelled by 
a literal $\pred(\tuple)\subS{\node_n}\dots\subS{\node_1}\vartheta$.}
\item{If a branch of the form $(h(\node_0),\dots,h(\node_n))$ is closed in $\otab$, then $(\node_0,\dots,\node_n)$ is closed in $\tab$.}
\end{enumerate}
Note that the mapping is not surjective in general ($\otab$ may be bigger than $\tab$). By \keyprop, if $\otab$ is closed then $\tab$ is also closed, which gives us the desired result. It is easy to check that applying the Separation rule on a node $\node_i$ in a branch $(\node_0,\dots,\node_n)$ in $\tab$ preserves \keyprop, provided it is applied  using a substitution $\theta$ (as defined in the Separation rule) that is more general than $\subS{\node_n}\circ\dots\circ\subS{\node_1}$.

The tableau $\tab$ is constructed inductively as follows. For the base case, we may take $\tab = \otab = \true$ and \keyprop trivially holds.

\textbf{Expansion:} The Expansion rule of ordinary tableaux allows one to expand the tableau by an arbitrary clause $\ordinary{\cl}_i$ of $\oclset$. We define the corresponding tableau for $\clset$ as the tableau expanded by $\cl_i$. The mappings $h$ and $\subS{\node}$ may be extended in a straightforward way so that \keyprop is preserved (the unique new node $\node$ in $\tab$ may be mapped to an arbitrary chosen new node in $\otab$).

\textbf{Closure:}
Assume that a branch in $\otab$ is closed by applying some substitution $\sigma$, using two literals $\pred(\tuple)$ and $\neg \pred(\tupleA)$, where $\sigma = \mgu(\tuple,\tupleA)$. If one of these literals do not occur in the image of a branch of $\tab$ then it is clear that the operation preserves \keyprop (except that the substitution $\vartheta$ is replaced by $\vartheta\sigma$), hence no further transformation is required on $\tab$. Otherwise, by  \keyprop, there is a branch $\abranch = (\node_0,\dots,\node_n)$ in $\tab$ where for every $i\in [0,n]$, $\node_i$ is labelled by $\lit_i =\Literal{\pred_i}{\tuple_i'}{\S_i}$, two numbers $j,k \in \mathbb{N}$ such that $\pred_j=\pred$, $\pred_k = \neg \pred$, $\tuple_j'\eta = \tuple$ and $\tuple_k'\eta=\tupleA$, with $\eta = \subS{\node_n}\dots\subS{\node_1}$.

By definition, $(\tuple_j',\tuple_k')$ is an $\eta$-\realizable \conflict. Let $\sigma'$ be the mgu of $\tuple_j'$ and $\tuple_k'$. If there is no \blocking variables for  $(\tuple_j',\tuple_k')$, then the Instantiation rule applies, replacing every variable $x\in \dom(\sigma')$ by $x\sigma'$, and the branch may be closed. By definition $\eta\vartheta\sigma$ is a unifier of $\tuple$ and $\tupleA$, hence $\eta\vartheta\sigma = \sigma'\theta'$, for some substitution $\theta'$. By definition, the co-domain of $\eta$ contains no {\evariable}s, thus $\theta' = \eta\vartheta'$, with $\dom(\vartheta') \subseteq \uVars$, and $\vartheta\sigma = \sigma'\vartheta'$.
The application of the rule preserves \keyprop, where $\vartheta$ is replaced by the substitution $\vartheta'$. Indeed, consider a node $\node$
in $\tab$, initially labelled by a literal $\Literal{\pred}{\tuple}{\S}$, where $h(\node)$ is labelled by $\pred(\tuple)\subS{\node_n}\dots\subS{\node_1}\vartheta$.
After the rule application, $\node$ is labelled by $\Literal{\pred}{\tuple\sigma'}{\S\sigma'}$, $h(\node)$ is labelled by $\pred(\tuple)\subS{\node_n}\dots\subS{\node_1}\vartheta\sigma$, and the substitutions $\subS{\node_i}$ are replaced by $\subS{\node_i}\sigma'$. Since $\dom(\sigma') \cap \eVars = \emptyset$, it is clear that $\sigma'\subS{\node_n}\dots\subS{\node_1}\sigma'\vartheta '=\subS{\node_n}\dots\subS{\node_1}\sigma'\vartheta '=\subS{\node_n}\dots\subS{\node_1}\vartheta\sigma$.

Otherwise, the set of \blocking variables is not empty, and since all {\evariable}s occurring in the tableau must be introduced in some node, there exists $l\in [1,n]$ such that $\dom(\S_l)$ contains a \blocking variable. According to Definition \ref{def:compact}, the Separation rule may be applied on $\node_l$ (it is easy to check that
all the application conditions of the rule are satisfied). In Case \ref{sep:inst} (of Definition \ref{def:compact}), $\theta$ is more general than  $\sigma'$ by definition, and since $\subS{\node_n}\circ\dots\circ\subS{\node_1}\vartheta\sigma$ is a unifier of $\tuple_j'$ and $\tuple_k'$, the mgu $\sigma'$ must be more general than $\subS{\node_n}\circ\dots\circ\subS{\node_1}$. In Case \ref{sep:elim}, we can take $\theta = \subS{\node_l}$, which is more general than $\subS{\node_n}\circ\dots\circ\subS{\node_1}$ (since the substitutions $\subS{\node}$ have disjoint domains). Thus the property \keyprop is preserved.

This operation is repeated until the set of \blocking variables is empty, which allows us to apply the Instantiation rule as explained before. The process necessarily terminates since each application of the Separation rule either increases the size of the tableau (either by adding new nodes, or by instantiating a variable by a non variable term), or does not increase the size of the tableau but strictly reduces the number of {\evariable}s. Furthermore, by \keyprop, the size of $\tab$ is smaller than that of $\otab$.
\end{proof}

\section{An Exponentially Compressed Tableau}

\label{sect:exp}

In this section, we will show that the presented method is able to compress tableaux by an exponential factor. This corresponds to an introduction of a single $\Pi_2$-cut\footnote{In a $\Pi_2$-cut, the cut formula is of the form $\forall x\exists yA$ where $A$ is a quantifier-free formula.} (see \cite{DBLP:journals/tcs/LeitschL18}). As a simplified measurement of the size of a tableau we consider the number of nodes. Let us consider the schema of \mclause sets $\clset^n\isdef\{\{\lit_1^n\} ,\{\lit_2,\lit_3\} ,\{\lit_4^n\}\}$ with
\begin{align*}
\lit_1^n= &\; \Literal{\predA}{\tuples{\alpha}}{[\tuples{\alpha}\backslash (x,f_1x)],\ldots ,[\tuples{\alpha}\backslash (x,f_nx)]} \\
\lit_2= &\; \Literal{\neg\predA}{\tuples{\beta}}{[\tuples{\beta}\backslash (x,y)]} \\
\lit_3= &\; \Literal{\predA}{\tuples{\gamma}}{[\tuples{\gamma}\backslash (x,fy)]} \\
\lit_4^n= &\; \Literal{\neg\predA}{\tuples{\delta}}{[\tuples{\delta}\backslash (fx_1,fx_2)],\ldots ,[\tuples{\gamma}\backslash (fx_{n-1},fx_n)]}
\end{align*}
where $\tuples{\alpha}=(\alpha_1,\alpha_2),\tuples{\beta}=(\beta_1,\beta_2),\tuples{\gamma}=(\gamma_1,\gamma_2)$, and $\tuples{\delta}=(\delta_1,\delta_2)$ for $n\in\mathbb{N}$. Then we can construct a closed tableau for $\clset^ n$ whose size is linear w.r.t.\ $n$:
{\small
\begin{center}
\begin{forest}
[{$\Literal{\predA }{(fx_1,\alpha_2^1)}{\{ [\alpha_2^1\backslash f_1fx_1],\ldots ,[\alpha_2^1\backslash f_nfx_1]\}}$}
 [{$\Literal{\neg \predA }{(fx_1,\alpha_2^1)}{\idset }$} ]
 [{$\Literal{\predA }{(fx_1,f\alpha_2^1)}{\idset }$} 
  [{$\Literal{\predA }{(f\alpha_2^1,\alpha_2^2)}{\{ [\alpha_2^2\backslash f_1f\alpha_2^1],\ldots ,[\alpha_2^2\backslash f_nf\alpha_2^1]\}}$}
   [{$\Literal{\neg \predA }{(f\alpha_2^1,\alpha_2^2)}{\idset }$} ]
   [{$\Literal{\predA }{(f\alpha_2^1,f\alpha_2^2)}{\idset }$} 
    [{$\tab$}, edge=dotted ] ] ] ] ]
\end{forest}
\end{center}
}
where $\tab$ is
{\small
\begin{center}
\begin{forest}
[{$\Literal{\predA }{(f\alpha_2^{n-2},\alpha_2^{n-1})}{\{ [\alpha_2^{n-1}\backslash f_1f\alpha_2^{n-2}],\ldots ,[\alpha_2^{n-1}\backslash f_nf\alpha_2^{n-2}]\}}$}
 [{$\Literal{\neg \predA }{(f\alpha_2^{n-2},\alpha_2^{n-1})}{\idset }$} ] 
 [{$\Literal{\predA }{(f\alpha_2^{n-2},f\alpha_2^{n-1})}{\idset }$} 
  [{$\Literal{\neg \predA }{(fx_1,f\alpha_2^1)}{\idset}$} ] 
  [{$\ldots$}, edge=dotted ]
  [{$\Literal{\neg \predA }{(f\alpha_2^{n-2},f\alpha_2^{n-1})}{\idset}$} ] ] ]
\end{forest}
\end{center}
}
The clause set schema $\clset^n$ is a simplified variant of the example in \cite[Section 3 and 9]{DBLP:journals/tcs/LeitschL18} and one can easily verify that an ordinary tableau method is of exponential size of $n$. Just consider the term instantiations which are necessary for an ordinary tableau:
\begin{align*}
&\; x\leftarrow\{ fx_1,ff_{i_1}fx_1,\ldots ,ff_{i_{n-2}}f\ldots f_{i_1}fx_1\vert i_1,\ldots ,i_{n-2}\in [1,n]\}\text{ in }\lit_1^n, \\
&\; (x,y)\leftarrow\{(t,f_it)\vert i\in [1,n]\land t\text{ is a substitution for }x\text{ in }\lit_1^n\}\text{ in }\lit_2\text{ and }\lit_3,\text{ and } \\
&\; (x_1,\ldots ,x_n)\leftarrow\{(x_1,f_{i_1}fx_1,\ldots ,f_{i_{n-1}}f\ldots f_{i_1}fx_1)\vert i_1,\ldots ,i_{n-1}\in [1,n]\}\text{ in }\lit_4^n.
\end{align*}
Obviously, $x_n$ in $\lit_4^n$ is substituted with $n^{n-1}$ terms. These correspond to the instantiations defined in \cite[Theorem 13]{DBLP:journals/tcs/LeitschL18}.

\section{Future Work}

From a practical point of view, algorithms and data-structures  have to be devised to apply the above rules efficiently, especially to identify {\conflict}s and \blocking substitutions in an incremental way. In the wake of this, experimental evaluations based on an implementation are reasonable. While the procedure is described for first-order logic, we believe that the same ideas could be profitably be applied to other logics, and even to other calculi, including saturation-based procedures. It would also be interesting to combine this approach with other techniques for reducing proof size, for instance variable splitting \cite{DBLP:journals/jsc/HansenAGW12}, or with techniques for the incremental construction of closures \cite{Giese01a}.

A current restriction of the calculus is that no {\evariable}s may occur above their introduction (see the condition on the Instantiation rule in Section \ref{sect:rules}). 
This restriction is essential for soundness: without it, one could for instance construct a closed tableau for the (satisfiable) set of {\mclause}s $\{ \{ \lit \}, \{ \lit' \} \}$, with $\lit = \Literal{\pred}{(x,\alpha)}{\{ [\alpha\backslash a, \alpha \backslash b] \}}$ and $\lit' = \Literal{\pred}{(\beta,y)}{\{ [\beta\backslash a, \beta \backslash b] \}}$, by replacing $x$ by $\alpha$
and $y$ by $\beta$. We think that this condition can be relaxed by defining an order over the {\evariable}s. This would yield a more flexible calculus, thus further reducing proof size. It would be interesting to know whether the exponential bound of Remark \ref{rem:bound} still holds for the relaxed calculus.
An ambitious long-term goal is to devise extensions of {\tableau}x with the same deductive power of cuts, i.e.\ enabling a non-elementary reduction of proof size.

\bibliography{LP_2018}
\bibliographystyle{abbrv}

\end{document}